\documentclass[10pt,journal,english,twocolumn]{IEEEtran} %
\usepackage{babel}
\usepackage[babel]{microtype}
\usepackage[babel]{csquotes}

\usepackage[T1]{fontenc}
\usepackage[svgnames]{xcolor}
\usepackage{graphicx}
\usepackage{tikz}
\usepackage{pgfplots}
\usepackage[caption=false,font=footnotesize]{subfig}

\usepackage{tabularx}
\usepackage{booktabs}

\usepackage{amsmath}
\usepackage{amsfonts}
\usepackage{amssymb}
\usepackage{amsthm}
\usepackage{bm}
\usepackage[binary-units=true]{siunitx}
\usepackage{mathtools}
\usepackage{dsfont}

\usepackage[math]{blindtext}
\usepackage{titlesec} %

\usepackage[backend=biber,url=false,citestyle=ieee,bibstyle=ieee-proc,dashed=false,isbn=false,doi=true,maxnames=10]{biblatex}
\bibliography{modernize-literature.bib}
\renewcommand*{\bibfont}{\small}

\usepackage{hyperref}
\hypersetup{
	pdfauthor={Karl-Ludwig Besser},
	pdftitle={Reliability and Latency Analysis for Wireless Communication Systems with a Secret-Key Budget},
	bookmarksnumbered=true,
	colorlinks=true,
	allcolors=.,
	urlcolor=MidnightBlue,
	pdflang={en-US},
}

\usepackage[acronyms,nomain,xindy,nowarn]{glossaries}
\makeglossaries
\newacronym{5g}{5G}{Fifth-Generation Wireless Systems}
\newacronym{ai}{AI}{Artificial Intelligence}
\newacronym{ann}{ANN}{Artificial Neural Network}
\newacronym{ask}{ASK}{amplitude-shift keying}
\newacronym[plural={AWGN}, \glsshortpluralkey={AWGN}]{awgn}{AWGN}{additive white Gaussian noise}

\newacronym{bec}{BEC}{binary erasure channel}
\newacronym{ber}{BER}{bit error rate}
\newacronym{bler}{BLER}{block error rate}
\newacronym{bpsk}{BPSK}{binary phase-shift keying}

\newacronym{cdf}{CDF}{cumulative distribution function}
\newacronym{cm}{CM}{channel model}
\newacronym{cnn}{CNN}{convolutional neural network}
\newacronym{csi}{CSI}{channel state information}
\newacronym{csit}{CSI\babelhyphen{nobreak}T}{channel-state information at the transmitter}
\newacronym{csir}{CSI\babelhyphen{nobreak}R}{channel-state information at the receiver}

\newacronym{ecc}{ECC}{error-correcting code}
\newacronym{ecdf}{ECDF}{empirical distribution function}
\newacronym{ee}{EE}{energy efficiency}

\newacronym{ff}{FF}{feed-forward}
\newacronym{ffnn}{FF-NN}{feed-forward neural network}
\newacronym{fft}{FFT}{fast Fourier transform}
\newacronym{fsk}{FSK}{frequency-shift keying}

\newacronym{gm}{GM}{Gaussian mixture}

\newacronym{iid}{i.i.d.}{independent and identically distributed}
\newacronym{il}{IL}{information leakage}
\newacronym{iot}{IoT}{internet of things}

\newacronym{lb}{LB}{lower bound}
\newacronym{los}{LoS}{line-of-sight}

\newacronym{mac}{MAC}{multiple access channel}
\newacronym{map}{MAP}{maximum a posteriori}
\newacronym{mc}{MC}{Monte Carlo}
\newacronym{mimo}{MIMO}{multiple-input multiple-output}
\newacronym{miso}{MISO}{multiple-input single-output}
\newacronym{ml}{ML}{machine learning}
\newacronym{mmwave}{mmWave}{millimeter wave}
\newacronym{mop}{MOP}{multi-objective programming problem}
\newacronym{mrc}{MRC}{maximum ratio combining}
\newacronym{msc}{MSC}{modulation and coding scheme}
\newacronym{mse}{MSE}{mean squared error}

\newacronym{nlos}{NLoS}{non-line-of-sight}
\newacronym{nn}{NN}{neural network}
\newacronym{nve}{NVE}{normalized validation error}

\newacronym{pdf}{PDF}{probability density function}
\newacronym{pwc}{PWC}{polar wiretap code}

\newacronym{qam}{QAM}{quadrature amplitude modulation}

\newacronym{ra}{RA}{rearrangement algorithm}
\newacronym{rbf}{RBF}{radial basis function}
\newacronym{relu}{ReLU}{rectified linear unit}
\newacronym{ris}{RIS}{reconfigurable intelligent surface}
\newacronym{rnn}{RNN}{recurrent neural network}

\newacronym{sgd}{SGD}{stochastic gradient descent}
\newacronym{sgp}{SGP}{secure goodput}
\newacronym{simo}{SIMO}{single-input multiple-output}
\newacronym{sinr}{SINR}{signal-to-interference-plus-noise ratio}
\newacronym{siso}{SISO}{single-input single-output}
\newacronym{sk}{SK}{secret-key}
\newacronym{skg}{SKG}{secret-key generation}
\newacronym{sm}{SM}{source model}
\newacronym{snr}{SNR}{signal-to-noise ratio}
\newacronym{svm}{SVM}{support vector machine}

\newacronym{tvd}{TVD}{total variation distance}
\newacronym{tx}{TX}{transmission}

\newacronym{uav}{UAV}{unmanned aerial vehicle}
\newacronym{ub}{UB}{upper bound}
\newacronym{urllc}{URLLC}{ultra-reliable low latency communication}

\newacronym{v2v}{V2V}{vehicle-to-vehicle}
\newacronym{v2x}{V2X}{vehicle-to-everything}

\newacronym{zoc}{ZOC}{zero-outage capacity}
\setacronymstyle{long-short}
\glsdisablehyper

\definecolor{plot0}{HTML}{004488}
\definecolor{plot1}{HTML}{DDAA33}
\definecolor{plot2}{HTML}{BB5566}
\definecolor{plot3}{HTML}{000000}
\definecolor{plot4}{HTML}{AAAAAA}

\DeclarePairedDelimiter{\abs}{\vert}{\vert}

\newcommand*{\inv}[1]{\ensuremath{#1^{-1}}}

\newcommand*{\diff}{\ensuremath{\mathrm{d}}}

\newcommand*{\e}{\ensuremath{\mathrm{e}}}

\DeclareMathOperator{\var}{var}
\DeclareMathOperator{\one}{\mathds{1}}

\newcommand*{\naturals}{\ensuremath{\mathds{N}}}

\newcommand*{\expect}[2][]{\ensuremath{\mathbb{E}_{#1}\left[#2\right]}}

\newcommand*{\berndist}{\ensuremath{\mathcal{B}}}

\let\rv\bm

\newcommand*{\simdist}{\ensuremath{\overset{d}{=}}}
\newcommand*{\X}{\ensuremath{\rv{X}}}
\newcommand*{\Xt}{\ensuremath{\rv{\tilde{X}}}}
\newcommand*{\Y}{\ensuremath{\rv{Y}}}
\newcommand*{\ratebob}{\ensuremath{{R_B}}}
\newcommand*{\rateeve}{\ensuremath{{R_E}}}
\newcommand*{\rateskg}[1][]{\ignorespaces
	\if\relax\detokenize{#1}\relax
	{\ensuremath{R_{SK}}}%
	\else
	{\ensuremath{R_{SK, #1}}}%
	\fi
}

\newcommand*{\claim}{\ensuremath{\rv{\xi}}}
\newcommand*{\income}{\ensuremath{\rv{\theta}}}
\newcommand*{\surplus}{\ensuremath{\rv{B}}}
\let\budget\surplus
\newcommand*{\netclaim}{\ensuremath{\rv{Z}}}
\newcommand*{\netsum}{\ensuremath{\rv{S}}}
\newcommand*{\probruin}{\ensuremath{\psi}}
\newcommand*{\probsurv}{\ensuremath{\bar{\psi}}}
\newcommand*{\initbudget}{\ensuremath{b_0}}

\newcommand*{\probtx}{\ensuremath{p}}
\newcommand*{\txindicator}{\ensuremath{\rv{P}}}

\newcommand*{\boundupper}{\ensuremath{\Psi}}

\newcommand*{\survtime}{\ensuremath{\tau}}
\newcommand*{\latency}{\ensuremath{\rv{T}}}

\newcommand*{\adjcoeff}{\ensuremath{r}}
\usetikzlibrary{positioning,calc,external}

\pgfplotsset{compat=newest}
\pgfplotscreateplotcyclelist{lineplot cycle}{ %
	{plot0, mark=*, thick, mark options=solid},
	{plot1, mark=triangle*, thick, mark options=solid},
	{plot2, mark=square*, thick, mark options=solid},
	{plot3, mark=diamond*, thick, mark options=solid},
	{plot4, mark=pentagon*, thick, mark options=solid},
}

\pgfplotsset{%
	betterplot/.style={
		width=.93\linewidth,
		height=.28\textheight,
		xlabel near ticks,
		ylabel near ticks,
		cycle list name=lineplot cycle,
		mark options=solid,
		xmajorgrids=true,
		xminorgrids=true,
		ymajorgrids=true,
		grid style={line width=.1pt, draw=gray!20},
		major grid style={line width=.25pt,draw=gray!30},
		legend cell align=left,
		legend style = {
			/tikz/every even column/.append style={column sep=0.33cm}
		},
	},
}

\newcommand{\todo}[2][]{\ignorespaces
	\if\relax\detokenize{#1}\relax
	{\color{red}[TODO: #2]}%
	\else
	{\color{red}[TODO (#1): #2]}%
	\fi
}

\definecolor{change}{HTML}{0096b8}

\theoremstyle{plain}%
\newtheorem{thm}{Theorem}
\newtheorem{cor}{Corollary}
\newtheorem{lem}{Lemma}

\theoremstyle{definition}
\newtheorem{prob}{Problem Statement}
\newtheorem*{prob*}{Problem Statement}

\theoremstyle{remark}
\newtheorem{rem}{Remark}
\newtheorem*{rem*}{Remark}

\newtheoremstyle{example}{\topsep}{\topsep}{}{}{\itshape}{.}{ }{}
\theoremstyle{example}
\newtheorem{example}{Example}
\newtheorem*{example*}{Example}

\addto\extrasenglish{%

}

\title{{Reliability and Latency Analysis for Wireless Communication Systems with a Secret-Key Budget}}
\author{%
	Karl-Ludwig Besser, \IEEEmembership{Member, IEEE}, Rafael F. Schaefer, \IEEEmembership{Senior Member, IEEE}, and\\H. Vincent Poor, \IEEEmembership{Life Fellow, IEEE}
	\thanks{Karl-Ludwig Besser was with the Institute for Communications Technology, Technische Universität Braunschweig, 38106 Braunschweig, Germany, and is now with the Department of Electrical and Computer Engineering, Princeton University,  Princeton, NJ 08544, USA (email: {karl.besser}@princeton.edu).
	Rafael F. Schaefer is with the Chair of Information Theory and Machine Learning, the BMBF Research Hub 6G-life, the Cluster of Excellence \enquote{Centre for Tactile Internet with Human-in-the-Loop (CeTI),} and the 5G Lab Germany, Technische Universit\"at Dresden, 01062 Dresden, Germany (e-mail: rafael.schaefer@tu-dresden.de).
	{H. Vincent Poor is with the Department of Electrical and Computer Engineering, Princeton University,  Princeton, NJ 08544, USA (email: {poor}@princeton.edu).}}
	\thanks{The work of K.-L. Besser is supported by the German Research Foundation (DFG) under grant BE\,8098/1-1.
	The work of R.~F. Schaefer is supported by the German Federal Ministry of Education and Research (BMBF) within the national initiative on 6G Communication Systems through the research hub 6G-life under Grant 16KISK001K as well as the 6G-ANNA project under Grant 16KISK103.
	The work of H.~V. Poor is supported by the U.S. National Science Foundation under Grants CNS-2128448 and ECCS-2335876.}
}

\begin{document}
\maketitle

\begin{abstract}\noindent\boldmath
	We consider a wireless communication system with a passive eavesdropper, in which a transmitter and legitimate receiver generate and use key bits to secure the transmission of their data.
	These bits are added to and used from a pool of available key bits.
	In this work, we analyze the reliability of the system in terms of the probability that the budget of available key bits will be exhausted.
	In addition, we investigate the latency before a transmission can take place.
	Since security, reliability, and latency are three important metrics for modern communication systems, it is of great interest to jointly analyze them in relation to the system parameters.
	In particular, we show under what conditions the system may remain in an active state indefinitely, i.e., never run out of available secret-key bits.
	The results presented in this work will allow system designers to adjust the system parameters in such a way that the requirements of the application in terms of both reliability and latency are met.
\end{abstract}
\glsresetall

\section{Introduction}\label{sec:introduction}

The next generation of mobile communication systems~6G, is expected to bring significant advances in terms of capacity, speed, and connectivity~\cite{Saad2020}.
However, with the increasing reliance on wireless networks for a variety of applications and the growing amount of sensitive information transmitted over these networks, the need for robust security measures becomes paramount~\cite{Nguyen2021}.
While cryptography is currently the most widely used technique to protect data transmissions from potential eavesdroppers, physical layer security provides an alternative solution~\cite{Porambage2021,Chorti2022}.
In particular, physical layer security allows a system to achieve a degree of security that is provable from an information theoretic viewpoint, rather than relying on the presumed impracticality of computational problems.
This is done by exploiting the physical properties of the communication channel between the transmitter and a legitimate receiver.
There are several ways to do this~\cite{Poor2017}.
In particular, in addition to using wiretap codes, the characteristics of the communication channel can be used to securely generate bits that are shared only by the transmitter and the legitimate receiver.
These bits can then act as a secret key in the form of a one-time pad to secure a message.

\Gls{skg} uses channel reciprocity to establish a common randomness between the transmitter and the legitimate receiver, which in turn can be used to distill the key bits~\cite{Li2019skg,Ren2011}.
Due to the physical nature of wireless propagation, the channel between transmitter and legitimate receiver is difficult to predict for outside observers.
Therefore, it is not possible for a potential eavesdropper to reconstruct the generated key bits, i.e., they can act as a shared key between transmitter and legitimate receiver.
For the specific implementation of \gls{skg}, various schemes have been proposed in the literature~\cite{KameniNgassa2017,Li2018skg,Aono2005,Sayeed2008}.

In this work, we consider a wireless communication system in which the transmitter and a legitimate receiver perform \gls{skg} to secure the transmission of their data.
The generated key bits are added to and used from a pool of available bits.
We analyze the reliability in terms of the probability that the budget of available key bits will be exhausted.
In addition, we investigate the latency before a transmission can take place.
For the analysis in this work, we leverage tools from ruin theory~\cite{Asmussen2010}.

Classical ruin theory addresses the problem of modeling the evolution of an insurance company's financial surplus and its risk of becoming insolvent~\cite{Asmussen2010}.
In the classical ruin-theoretic model, also known as the Cram\'er-Lundberg model~\cite[Chap.~IV.1]{Asmussen2020}, the insurance company experiences the following two cash flows.
On the one hand, it receives a constant stream of income from insurance premiums.
On the other hand, random claims arrive according to a Poisson process.
This leads to problems such as calculating the probability of ruin, i.e., the probability that the total surplus becomes negative.

Tools from ruin theory have been applied in previous works to various problems in the broad area of wireless communications.
Along the traditional lines of considering monetary quantities, the probability of financial ruin for network-sharing arrangements is considered in~\cite{egan_ruin_2017}.
Other applications include resource allocation, e.g,. spectrum sharing~\cite{manzoor_ruin_2019,Htike2020}, user association~\cite{Kim2021ruin,manzoor_ruin_2021}, and power allocation~\cite{manzoor_ruin_2021}.
In~\cite{manzoor_ruin_2021}, a \gls{uav}-assisted cellular network is considered.
The authors use the traditional ruin-theoretic model of constant income and random claims to describe the available energy of the \glspl{uav}.
The \glspl{uav} have only a finite amount of stored energy, which is consumed when transmitting data.
The energy consumption corresponds to the claims in the ruin-theoretic model.
At the same time, the battery is recharged by harvesting solar energy at a constant rate, which corresponds to the income stream in the traditional insurance model.

While we also use the basic structure of the traditional ruin-theoretic model of income and claim, we consider an adapted model in this work.
In particular, the income rate is not constant but a random variable, since the quality of the wireless channels, and thus the \gls{skg} rate, varies randomly.
Additionally, we assume two different scheduling schemes for generating new key bits and transmitting messages, which affects the way that claims arrive.

To the best of the authors' knowledge, this work is the first to consider a communication system with a secret-key budget and the goal of this work is to provide a framework for a theoretical analysis of such systems.
Therefore, the results derived in this work are applicable to a variety of scenarios and, in particular, they are not limited to a specific underlying \gls{skg} scheme.
Our main contributions and the outline of the manuscript are summarized as follows.
\begin{itemize}
	\item
	We formulate a model for a communication system with a \gls{sk} budget, in which new \gls{sk} bits are generated and then used for securely transmitting messages to a legitimate receiver	(\autoref{sec:system-model}).
	\item
	For this model, we first consider a scheme where \gls{skg} and transmission alternate (\autoref{sec:fixed-time-scheme}).
	Reliability is analyzed in terms of both the outage probability and latency.
	It is shown that the system will eventually run out of key bits to securely transmit messages with probability one (\autoref{cor:prob-ultimate-ruin-deterministic}).
	\item
	In addition, we investigate an alternative operating scheme where messages arrive randomly with probability~$\probtx$ in each time slot and \gls{skg} is performed whenever no secure message is to be transmitted (\autoref{sec:random-arrival-time-scheme}).
	For this setting, it may happen that the communication system can operate indefinitely.
	We characterize the range of $\probtx$ for which this is possible (\autoref{cor:expected-net-claim-random}).
\end{itemize}
The source code to reproduce all presented results and simulations is made publicly available at~\cite{BesserGithub}.
{The provided source code can also be freely adapted to custom numerical examples.}

\textit{Notation:}
Random variables are denoted in capital boldface letters, e.g., $\X$, and their realizations in small letters, e.g., $x$.
We use $F_{\X}$ and $f_{\X}$ for the probability distribution and its density, respectively.
The expectation is denoted by $\mathbb{E}$ and the probability of an event by $\Pr$.
The relation $\rv{X}\simdist\rv{Y}$ holds for two random variables $\rv{X}$ and $\rv{Y}$ when they are distributed according to the same probability distribution.
The Bernoulli distribution with mean~$p$ is denoted as $\berndist(p)$.
The indicator function is written as~$\one$.
\section{Preliminaries and Background Knowledge}\label{sec:background}

The two primary groups of themes that are touched upon in this work are secret-key generation and ruin theory.
In order to improve the flow of reading, we will give a brief introduction to these areas in this section.

\subsection{Secret-Key Generation}\label{sub:intro-skg}
One way to achieve perfect information-theoretic secrecy is the use of a one-time pad~\cite{Shannon1949}.
This one-time pad needs to consist of key bits that are only known to the legitimate parties.
Secret-key generation is a way for two legitimate communication parties to agree on such secret key bits through exploiting the physical properties of their communication channel~\cite{Bloch2021}.
In order to achieve this, multiple models and algorithms have been proposed and analyzed in the literature for various communication scenarios~\cite{Furqan2016,Letafati2021globecom,Gao2021,Linh2023}.
This includes static environments~\cite{Aldaghri2020}, quasi-static fading channels~\cite{Renna2013}, and fast-fading channels with correlated channels~\cite{Zorgui2016,Besser2020wsa}.

The concept behind secret-key agreement is that Alice and Bob have access to shared randomness, which can be used to extract identical bits to serve as key bits~\cite{Maurer1993}.
To correct errors when observing the randomness and protect against eavesdroppers, Alice and Bob can exchange messages across a public channel.
The most notable models for secret-key agreement are the \gls{cm} and \gls{sm}~\cites[Chap.~4]{Bloch2011}{Lai2014}, each differing in the generation of the random observations.
In the \gls{sm}~\cite{Wallace2010}, both legitimate nodes and the eavesdropper have access to some source of common randomness modeled by a joint distribution.
In the \gls{cm}, the randomness originates from transmissions over a noisy wiretap channel~\cite{Ahlswede1993}.
In both models, the public channel is used to generate the secret key bits from the observed randomness at the legitimate nodes, leaving the eavesdropper with no information about it.

In this work, we consider a communication system with a secret-key budget.
As described above, the legitimate nodes perform \gls{skg} and agree on secret key bits using any \gls{skg} scheme.
They then append the newly generated bits to previously generated key bits.
This way, a pool of available \gls{sk} bits is built up at both legitimate communication parties.
Whenever a secure transmission of a message with length~$n$ takes place, the oldest $n$~\gls{sk} bits are used as a one-time pad to encrypt the message.
Since only the legitimate transmitter and receiver know the key bits, the transmission is information-theoretically secure.
Due to the nature of one-time pads, the \gls{sk} bits used for the transmission can only be used once and are therefore removed from the list of available key bits.

\subsection{Ruin Theory}\label{sub:intro-ruin-theory}
Ruin theory originated in economics and actuarial science, where it is used to analyze the solvency of an insurance company~\cite{Asmussen2020}.
The basic idea is that an insurance company deals with two opposing cash flows simultaneously.
On the one hand, they receive money in the form of premiums paid by customers.
On the other hand, there occur claims that are paid by the insurance company.

In the traditional model, the premiums arrive at a constant (positive) rate, while the claims arrive randomly according to a Poisson process.
The main quantity of interest is the probability that the insurance company will go bankrupt.

For this section, let $\tau = \inf \{t \geq 0 \;|\; X(t) \leq 0 \}$ be the time of ruin, i.e., the first time~$t$ at which the aggregate surplus~$X$ falls to zero~\cite{Dickson2016}.
This defines the probability of ultimate ruin~$\probruin_{\infty}$ as the probability that bankruptcy will eventually occur,
\begin{equation*}
	\probruin_{\infty} = \Pr(\tau < \infty)\,.
\end{equation*}
Similarly, the probability of ruin in finite time~$\probruin_{t}=\Pr(\tau \leq t)$ is given as the probability of ruin before time~$t$.

Ruin theory has been extensively studied in the literature of mathematical finance and actuarial science, where many different problems have been discussed and explicitly solved, including expressions for finite-time ruin probability when considering specific claim distributions~\cite{Chan2006,Picard1997} or when considering interest rates~\cite{Cai2002,Wang2006interest}.
For a more detailed overview of the topic, we refer the reader to~\cite{Asmussen2020,Dickson2016}.

Due to the differences between financial systems and communication systems, we cannot simply apply all existing results from the literature.
Instead, we will take some of the above ideas and definitions from the area of ruin theory as a basis and adapt them to the problem considered in this work.
The system model and the exact problem formulation are discussed in the following section.
\section{System Model and Problem Formulation}\label{sec:system-model}
Throughout this work, we consider a wiretap channel, where a transmitter (Alice) wants to securely transmit data to a legitimate receiver (Bob).
The transmission is overheard by a passive eavesdropper (Eve).
Both the channel between Alice and Bob and the channel between Alice and Eve are {quasi-static fading channels} with \gls{awgn} at the receivers and \glspl{snr}~$\X$ and $\Y$, respectively~{\cite[Sec.~5.2]{Bloch2011}}.
The \gls{snr} values are assumed to be independent random variables, which also change independently over time.
The achievable rates to Bob and Eve at time~$t$ are given as
\begin{align}
	\rv{\ratebob}_{\!,t} &= \log_2(1+\X_t)\\
	\rv{\rateeve}_{\!,t} &= \log_2(1+\Y_t)\,,
\end{align}
respectively, where $\X_t$ and $\Y_t$ denote the \gls{snr} values of the main channel and eavesdropper's channel at time~$t$.
{According to the quasi-static model}, we assume that the channels remain constant for the transmission of one codeword and the time index denotes this (discrete) time slot~\cite[Sec.~5.2]{Bloch2011}.

To securely transmit messages to Bob, Alice uses key bits for encryption.
These key bits are generated by the standard \gls{skg} procedures for wiretap channels~\cite[Chap.~4]{Bloch2011} and stored in a pool of available key bits.
They are then used as a one-time pad to secure message bits, i.e., each \gls{sk} bit can be used exactly once and is then removed from the budget of available key bits.

The \gls{skg} rate in time slot~$i$ for the considered model is given by~\cite[Sec.~5.1]{Bloch2011} %
\begin{equation}
	\rv{R}_{\rv{SK},i} = \income_i = \log_2\left(\frac{1+\X_i+\Y_i}{1+\Y_i}\right)\,.
\end{equation}
This corresponds to the income of \gls{sk} bits to the budget in time slot~$i$.
The number of key bits that are required to encrypt a transmission in time slot~$j$ is determined by the transmission rate from Alice to Bob over the respective channel in time slot~$j$
\begin{equation}
	\rv{\ratebob}_{\!,j} = \claim_j = \log_2\left(1 + \Xt_j\right)\,,
\end{equation}
where $\Xt_j$ denotes the \gls{snr} during transmission.
We assume that $\Xt$ is distributed according to the same distribution as $\X$, i.e., $\Xt\simdist\X$; however, they are assumed to be independent.
For simplicity, we assume that the transmit power is the same for both \gls{skg} and data transmission.
However, this could further be optimized in future work.

Therefore, the total number of available \gls{sk} bits at time~$t$ is given by the difference between the total number of generated key bits and the total number of used key bits up to that time slot.
The total number of \gls{sk} bits available at the end of time slot~$t$ is denoted as $\budget_t$.
In addition to generating new key bits, we assume that the system starts with an initial budget~$\initbudget>0$.

An outage is defined as the event that the \gls{sk} budget is exhausted at the end of time slot~$t$, i.e., $\budget_t\leq 0$.
Based on this, we define the \emph{survival probability~$\probsurv_{t}$} as the probability that the communication will last until time slot~$t$, i.e., that $\budget_t>0$ holds for all time slots $i\leq t$,
\begin{equation}\label{eq:def-prob-surv}
	\probsurv_{t}(\initbudget) = \Pr\left(\min_{1\leq i \leq t}\budget_i > 0\right)\,.
\end{equation}
Similarly, we define the \emph{outage probability~$\probruin_{t}$} as its complement
\begin{equation}\label{eq:def-prob-ruin}
	\probruin_{t}(\initbudget) = 1-\probsurv_{t}(\initbudget)\,.
\end{equation}

\begin{figure}
	\centering
	\begin{tikzpicture}%
	\begin{axis}[%
		betterplot,
		width=.98\linewidth,
		xlabel={Time Slot~$t$},
		ylabel={Available \Gls{sk} Bits~$\budget_t$},
		xmin=0,
		xmax=30,
		ymin=0,
		ymax=12,%
		ytick={0, 10},
		yticklabels={$0$, $\initbudget$},
		]
		\addplot+[restrict x to domain=0:10] table [x=t,y=b] {data/samples_budget_illustration.dat};
		\addplot+[restrict x to domain=10:20] table [x=t,y=b] {data/samples_budget_illustration.dat};

		\pgfplotsset{cycle list shift=-2};
		\addplot+[restrict x to domain=20:30] table [x=t,y=b] {data/samples_budget_illustration.dat};
		
		\addplot[very thick, dashed, no markers] coordinates {(10, 0) (10, 12)};
		\addplot[very thick, dashed, no markers] coordinates {(20, 0) (20, 12)};
		
		\draw[fill=plot0,draw=none,opacity=.1] (axis cs: 0,0) rectangle (axis cs: 10, 12);
		\draw[fill=plot1,draw=none,opacity=.1] (axis cs: 10,0) rectangle (axis cs: 20, 12);
		\draw[fill=plot0,draw=none,opacity=.1] (axis cs: 20,0) rectangle (axis cs: 30, 12);
		
		\node[fill=white,font=\small] at (axis cs: 5, 11) {Active State};
		\node[fill=white,font=\small] at (axis cs: 15, 11) {{Recharge State}};
		\node[fill=white,font=\small] at (axis cs: 25, 11) {{Active State}};

		\draw[|<->|, thick] (axis cs:10,1) -- (axis cs:20,1);
		\node[fill=white, font=\small] (latency) at (axis cs:15,1) {Latency~$\latency$};
	\end{axis}
\end{tikzpicture}
	\caption{%
		Exemplary illustration of the temporal progress of the number of available \gls{sk} bits~$\budget_t$.
		During the {active state}, both \gls{skg} and transmission are performed.
		Once the budget is exhausted, the system switches to a {recharge state}, where only new key bits are generated until a certain threshold~$\initbudget$ is reached.
		The latency~$\latency$ is defined as the number of time slots between two {active states}.
		In the shown example, we have $\latency=10$.
	}
	\label{fig:model-budget-illustration}
\end{figure}
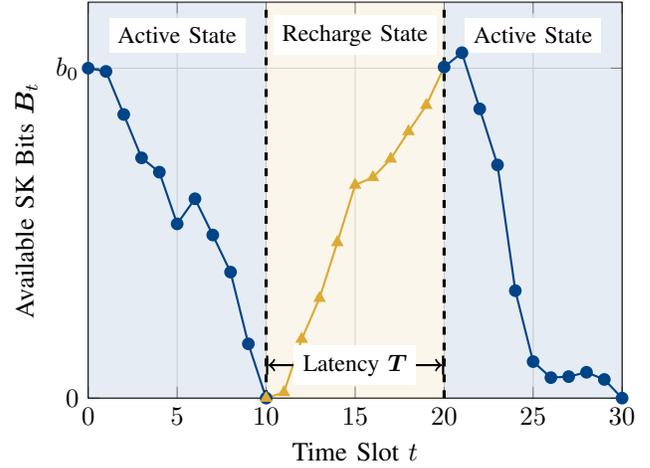

As long as there are still key bits available, i.e., $\budget_t > 0$, we say that the system is in the {\emph{active state}}.
Whenever the number of available \gls{sk} bits drops to zero, the system enters a \emph{recharge phase}, where no data is transmitted and only new \gls{sk} bits are being generated until the initial amount~$\initbudget$ is reached again.
At this point, the system switches back to the {active state}, where both \gls{skg} and data transmission occur, as described above.
An exemplary illustration of this system model can be found in \autoref{fig:model-budget-illustration}.

The duration between two {active states}, in which data transmission takes place, defines the latency~$\latency$ of the system.
It is given as the first time slot within the {recharge phase} in which the budget reaches the specified threshold~$\initbudget$, i.e.,
\begin{equation}\label{eq:def-latency}
	\latency(\initbudget) = \inf \left\{t\geq 0 \;\middle|\; \sum_{i=t_0}^{t}\income_i \geq \initbudget\right\} - t_0 + 1 \,,
\end{equation}
where $t_0$ denotes the first time slot of the recharge phase.
Note that we can set $t_0=1$ without loss of generality if the distribution of~$\income_i$ is stationary over time.
Recall that during the {recharge state} only \gls{skg} takes place and the number of \gls{sk} bits available in a specified time slot is given by the sum of the bits generated up to that time slot.
Since the number of generated bits~$\income_i$ in each time slot~$i$ is random, the latency~$\latency$ is also a random variable.

An overview of the most commonly used variable notation can be found in \autoref{tab:notation}.

The exact nature of the random process~$\budget_t$ depends on the communication scheme of the system.
In this work, we will discuss two different schemes.
First, we will consider a deterministic timing scheme in \autoref{sec:fixed-time-scheme}, where in each time slot both \gls{skg} and data transmission are performed.
Second, we analyze a different model in \autoref{sec:random-arrival-time-scheme}, where data transmissions occur randomly and \gls{skg} is performed when no data is to be transmitted.

\begin{table*}
	\renewcommand*{\arraystretch}{1.25}
	\centering
	\caption{Definitions of the Most Commonly Used Variables}
	\label{tab:notation}
	\begin{tabularx}{\linewidth}{lX}
		\toprule
		$\X_t$ & \Gls{snr} of the main channel (Alice to Bob) during the \gls{skg} phase in time slot~$t$\\
		$\Xt_t$ & \Gls{snr} of the main channel during data transmission in time slot~$t$\\
		$\Y_t$ & \Gls{snr} of the eavesdropper channel (Alice to Eve) during the \gls{skg} phase in time slot~$t$\\
		$\surplus_t$ & Amount of available \gls{sk} bits in time slot~$t$ \\ %
		$\income_t$ & \Gls{skg} rate in time slot~$t$ \\
		$\claim_t$ & Number of used \gls{sk} bits for transmission in time slot~$t$ \\
		$\netclaim_t$ & Net usage of \gls{sk} bits in time slot~$t$ \\
		$\netsum_t = \sum_{i=1}^{t}\netclaim_i$ & Accumulated net usage of \gls{sk} bits up to time slot~$t$ \\
		$\initbudget$ & Amount of initially available \gls{sk} bits \\
		$\probruin_t(\initbudget)$ & Outage probability until time~$t$ with initial budget~$\initbudget$ \\
		$\probsurv_{t}(\initbudget) = 1-\probruin_t(\initbudget)$ & Survival probability until time~$t$ with initial budget~$\initbudget$\\
		$\initbudget^{\survtime}(\varepsilon)$ & Required initial budget to survive at least $\survtime$~time slots with an outage probability of at most~$\varepsilon$\\
		$\latency$ & Latency between two {active} system states\\
		\bottomrule
	\end{tabularx}
\end{table*}

\subsection{Problem Formulation}
Based on the introduction of the system model above, two immediate problems arise, which will be explicitly stated in the following.

\begin{prob}\label{prob:ruin-probabilities}
	What is the outage probability~$\probruin_t(\initbudget)$ for given distributions of the channels and an initial budget~$\initbudget$?
\end{prob}

\begin{prob}
	What is the latency-reliability tradeoff for the described system? %
	In particular, given an application requirement that the transmission needs to last at least $\survtime$ time slots with an outage probability of at most $\varepsilon$, we can determine the minimum required initial budget~$\initbudget^{\survtime}(\varepsilon)$ by the solution of \autoref{prob:ruin-probabilities}.
	However, in order to generate this initial amount of \gls{sk} bits, we need a certain number of time slots before any transmission can start, which defines the latency~$\latency$ of the system.
\end{prob}

Both problems will be analyzed for the deterministic timing scheme and the random timing scheme in \autoref{sec:fixed-time-scheme} and \autoref{sec:random-arrival-time-scheme}, respectively.
\section{Deterministic Scheme}\label{sec:fixed-time-scheme}
As a first scheduling scheme for the {active state}, we consider a deterministic scheme where \gls{skg} and \gls{tx} alternate.
An illustration can be found in \autoref{fig:model-fixed-time-scheme}.
The channels are assumed to vary independently between each \gls{skg} and \gls{tx} block.
However, they remain constant for the entire duration of each single block.
Throughout this section, we refer to the combination of one full cycle of an \gls{skg} block followed by a transmission block as one time slot, i.e., each time slot consists of one \gls{skg} block and one \gls{tx} block, cf.~\autoref{fig:model-fixed-time-scheme}.
\begin{figure}
	\centering
	\begin{tikzpicture}%
	\begin{axis}[
		scale only axis,
		width=.9\linewidth,
		height=.13\textheight,
		axis lines=middle,
		xlabel={Time Slot~$t$},
		xlabel near ticks,
		x tick label as interval,
		ylabel near ticks,
		ylabel={Channel Gain},
		yticklabels={},
		ymajorticks=false,
		xmin=.99,
		xmax=5.2,
		xtick distance=1,
		ymin=0,
		ymax=1.2,
		axis on top,
		]
		\draw[plot2,thick,fill=plot2,fill opacity=.3, draw opacity=0] foreach \n in {1, ..., 5}{(\n,0) rectangle (\n+0.5,1)};
		\draw[plot3,thick,fill=plot3,fill opacity=.3, draw opacity=0] foreach \n in {1, ..., 5}{(\n+0.5,0) rectangle (\n+1,1)};
		
		\draw[black,very thick,dashed] foreach \n in {2, ..., 5}{(\n, 0) -- (\n, 2)};
		\draw foreach \n in {1, ..., 5}{(\n+0.25, 1) node[anchor=south] {SKG}};
		\draw foreach \n in {1, ..., 5}{(\n+0.75, 1) node[anchor=south] {TX}};
		
		\addplot[black, ultra thick, const plot] table[x=x, y=y] {
			x	y
			1	0.25
			1.5	0.5
			2	0.7
			2.5	0.3
			3	.1
			3.5	.6
			4	.5
			4.5	.4
			5	.8
			5.5	.8
		};
	\end{axis}
\end{tikzpicture}
	\caption{%
		Illustration of the scheduling with a deterministic scheme {during the active state}.
		In each time slot, there is an \gls{skg} block followed by a \gls{tx} block.
		While the channels are assumed to remain constant for each individual phase, they change independently between the \gls{skg} and \gls{tx} blocks.
	}
	\label{fig:model-fixed-time-scheme}
\end{figure}
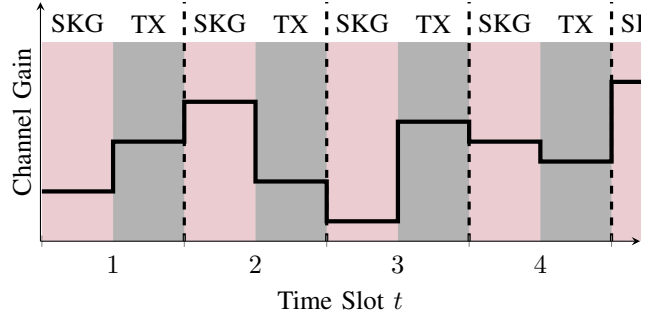

For this scheme, in each time slot~$i$, the amount~$\income_i$ corresponding to the \gls{skg} rate is added to the budget of available \gls{sk} bits, while $\claim_i$~bits are removed from it during the \gls{tx} block.
Thus, the total number of \gls{sk} bits available at the end of time slot~$t$ is given as
\begin{align}
	\surplus_t 
	&= \initbudget + \sum_{i=1}^{t} \income_i - \sum_{i=1}^{t} \claim_i\\
	&= \initbudget - \sum_{i=1}^{t} \netclaim_i\\
	&= \initbudget - \netsum_t\label{eq:def-budget-fixed}
\end{align}
where $\netclaim_i = \claim_i - \income_i$ describes the net usage of \gls{sk} bits in time slot~$i$ and 
\begin{equation}\label{eq:def-net-sum}
	\netsum_{t}=\sum_{i=1}^{t}\netclaim_i
\end{equation}
gives the accumulated net usage until time~$t$.

\subsection{Reliability Analysis}\label{sub:reliability-fixed}
First, we start with a reliability analysis for the scheme described above.
The first important observation is that the system loses \gls{sk} bits in each time slot on average, which is formalized in the following \autoref{lem:expected-net-claim}.
\begin{lem}[{Average Net Usage -- Deterministic Scheme}]\label{lem:expected-net-claim}
	Consider the described communication system in the {active state} with the deterministic timing scheme.
	The expected value of the net usage~$\netclaim_i = \claim_i - \income_i$ is positive, i.e., $\expect{\netclaim_i}>0$.
	Thus, the system's \gls{sk} budget reduces on average in every time slot.
\end{lem}
\begin{proof}
	The proof can be found in \autoref{app:proof-lem-avg-net-usage-fixed}.
\end{proof}

An important consequence that follows from \autoref{lem:expected-net-claim} is stated in the following.
\begin{cor}[{Probability of Ultimate Ruin -- Deterministic Scheme}]\label{cor:prob-ultimate-ruin-deterministic}
	The secret-key budget of the system in the {active state} using the deterministic timing scheme will almost surely be exhausted in finite time.%
\end{cor}
\begin{proof}
	Since $\netclaim_i$ has a positive expected value, cf.~\autoref{lem:expected-net-claim}, the sum~$\netsum_{t}$ forms a random walk with (positive) drift.
	From \cite[{Chap.~XII.2, Thm.~1 and 2}]{Feller1991}, it follows that $\netsum_{t}$ drifts to $+\infty$ with probability one.
	Thus, $\budget_t=\initbudget-\netsum_{t}$ will become negative with probability one.
\end{proof}

Having established that the secure communication cannot last indefinitely, it is of interest to quantify the outage probability at any given time slot~$t$.
This outage probability is defined according to \eqref{eq:def-prob-ruin}.
The complementary probability is the survival probability~$\probsurv_{t}$ from \eqref{eq:def-prob-surv}, i.e., the probability that the system will remain in the {active state} until time slot~$t$.
The survival probability~$\probsurv_{t}$ can be described by the following recursive relation~\cite{DeVylder1988,Willmot1993}:
\begin{equation}\label{eq:prob-surv-recursive}
	\probsurv_{t+1}(b) = \int_{-\infty}^{b}\probsurv_{t}(b-s)\diff{F_{\netclaim}(s)}\,,
\end{equation}
where $F_{\netclaim}$ describes the common distribution function of the \gls{iid} net claims~$\netclaim_i$.
The intuition behind this equation is that, in order to survive until time slot~$t+1$, the system needs to first survive until time~$t$.
Additionally, the net usage~$\netclaim_{t+1}$ of the next time slot needs to be small enough to leave a positive budget remaining.

\begin{rem*}
	The recursive relation \eqref{eq:prob-surv-recursive} can also be derived from the standard time-discrete ruin theoretic model without the income process~\cite{Dickson2016}.
\end{rem*}

\begin{example*}
	To illustrate \eqref{eq:prob-surv-recursive}, we solve the first steps explicitly.
	For the initial state of the system at time~$t=0$, the budget~$\budget_0$ is equal to the initial budget~$\initbudget$.
	The initial survival probability~$\probsurv_{0}(\initbudget)$ is therefore a step function, where the step from zero to one occurs at~$\initbudget$.
	Thus, for the next time step, we obtain from \eqref{eq:prob-surv-recursive} that
	\begin{equation*}
		\probsurv_{1}(\initbudget) = \int_{-\infty}^{\initbudget}\diff{F_{\netclaim}}(s) = F_{\netclaim}(\initbudget)\,,
	\end{equation*}
	i.e., the probability of surviving the first time slot is equal to the probability that the net claim~$\netclaim_1$ is less than the initial available budget~$\initbudget$.
	This relation can now be applied recursively to calculate the survival probabilities~$\probsurv_{t}$ for all time slots~$t$.
\end{example*}

However, the recursive equation in \eqref{eq:prob-surv-recursive} can be difficult to solve for general distributions of the net usage~$F_{\netclaim}$.
We are therefore interested in finding an efficient way to compute it numerically.
It is clear that \eqref{eq:prob-surv-recursive} is an integrodifference equation, which can be solved numerically by several different methods.
For a detailed treatise, we refer the interested reader to~\cite{Lutscher2019}.
Throughout the following, we use the \gls{fft} method to calculate~$\probsurv_{t}(\initbudget)$ according to \cite[Chap.~8]{Lutscher2019}.
Our implementation is made publicly available in \cite{BesserGithub}.

Even though \eqref{eq:prob-surv-recursive} can be efficiently solved numerically, it still requires a recursive calculation up to the desired time step~$t$.
Therefore, we provide an easy-to-calculate upper bound on the outage probability~$\probruin_{t}$ in the following theorem.
While this is only a (loose) worst-case bound, it can be easily computed without recursion for any given time slot~$t$.

\begin{thm}[{Worst-Case Bounds of the Ruin Probability -- Deterministic Scheme}]\label{thm:bound-upper-ruin-prob-fix}
	Consider the described communication system employing the deterministic timing scheme.
	The outage probability~$\probruin_{t}(\initbudget)$ can be upper bounded by
	\begin{equation}
		\probruin_{t}(\initbudget) \leq \boundupper_t(\initbudget) < \hat{\boundupper}_t(\initbudget)
	\end{equation}
	with
	\begin{equation}
		\boundupper_t(\initbudget) = \frac{\expect{\max\left(\netsum_t, 0\right)}}{\initbudget}
	\end{equation}
	and
	\begin{equation}\label{eq:bound-upper-loose-ruin-prob-fix}
		\hat{\boundupper}_t(\initbudget) = \frac{\sqrt{\var(\netsum_t)+\expect{\netsum_t}^2}}{\initbudget}\,.
	\end{equation}
\end{thm}
\begin{proof}
	The proof can be found in \autoref{app:proof-thm-up-bound-outage-prob-fixed}.
\end{proof}

\begin{rem}
	Since we assume that all channel realizations are mutually independent, we can simplify \eqref{eq:bound-upper-loose-ruin-prob-fix} to
	\begin{equation}
		\hat{\boundupper}_t(\initbudget) = \frac{\sqrt{\sum_{i=1}^{t}\var(\netclaim_i)+\left(\sum_{i=1}^{t}\expect{\netclaim_i}\right)^2}}{\initbudget}\,.
	\end{equation}
	If we on top of that assume that the distribution of the channel gains does not vary over time, i.e.,
	$\netclaim_1\simdist\netclaim_2\simdist\cdots{}\simdist\netclaim_i$ for all $i$, 
	we can further simplify the expression to
	\begin{equation}
		\hat{\boundupper}_t(\initbudget) = \frac{\sqrt{t\var(\netclaim_1)+t^2\expect{\netclaim_1}^2}}{\initbudget}\,.
	\end{equation}
\end{rem}

\begin{example}[Rayleigh Fading]\label{ex:fixed-time-rayleigh}
	We now illustrate the reliability analysis with a numerical example.
	For this purpose, we assume that both Bob and Eve experience Rayleigh fading, i.e., the channel gains~$\X_i$ and $\Y_i$ are distributed according to an exponential distribution.
	The channel realizations are \gls{iid} over time and independent in space.
	The average \gls{snr} of Bob's channel is set to $\SI{20}{\dB}$ while Eve's average \gls{snr} is $\SI{10}{\dB}$.
	With these parameters, the average net usage of \gls{sk} bits per time slot is about $\expect{\netclaim_i}=\SI{2.58}{\bit}$.
	In \autoref{fig:outage-prob-rayleigh-fixed-time}, the outage probability~$\probruin_{t}(\initbudget)$ is shown over time for different initial \gls{sk} budgets~$\initbudget$.
	The solid lines indicate the numerically calculated values from \eqref{eq:prob-surv-recursive}, while the markers are obtained from \gls{mc} simulations with $10^6$ samples.
	All of the calculations and simulations can be reproduced using the source code provided in~\cite{BesserGithub}.
	
	As expected from \autoref{cor:prob-ultimate-ruin-deterministic}, all of the outage probabilities approach \num{1} over time.
	However, increasing the initially available number of \gls{sk} bits~$\initbudget$ decreases the outage probability, i.e., the system stays longer in the {active state} for a given outage probability.
	The upper bounds from \autoref{thm:bound-upper-ruin-prob-fix} are not shown in \autoref{fig:outage-prob-rayleigh-fixed-time} because they are loose for this particular example.
	For $\initbudget=50$ at time~$t=15$, the bound is about $\hat{\boundupper}_{15}(50)=0.80$, while both \eqref{eq:prob-surv-recursive} and \gls{mc} simulations yield the actual outage probability of about $\probruin_{15}(50)=0.11$.
	\begin{figure}
		\centering
		\begin{tikzpicture}%
	\begin{axis}[
		betterplot,
		xlabel={Time Step $t$},
		ylabel={Outage Probability $\probruin_{t}(\initbudget)$},
		xmin=1,
		xmax=30,
		ymode=log,
		ymin=1e-7,
		ymax=2,
		legend pos=south east,
		]
		
		\addplot+[no marks, legend image post style={mark=*}] table[x=time,y=ide] {data/ruin-cdf-time-b5.0.dat};
		\addlegendentry{$\initbudget=5$};
		
		\addplot+[no marks, legend image post style={mark=triangle*}] table[x=time,y=ide] {data/ruin-cdf-time-b10.0.dat};
		\addlegendentry{$\initbudget=10$};
		
		\addplot+[no marks, legend image post style={mark=square*}] table[x=time,y=ide] {data/ruin-cdf-time-b20.0.dat};
		\addlegendentry{$\initbudget=20$};
		
		\addplot+[no marks, legend image post style={mark=diamond*}] table[x=time,y=ide] {data/ruin-cdf-time-b50.0.dat};
		\addlegendentry{$\initbudget=50$};

		\pgfplotsset{cycle list shift=-4};
		\addplot+[only marks] table[x=time,y=mc] {data/ruin-cdf-time-b5.0.dat};
		\addplot+[only marks] table[x=time,y=mc] {data/ruin-cdf-time-b10.0.dat};
		\addplot+[only marks] table[x=time,y=mc] {data/ruin-cdf-time-b20.0.dat};
		\addplot+[only marks] table[x=time,y=mc] {data/ruin-cdf-time-b50.0.dat};
	\end{axis}
\end{tikzpicture}
		\caption{%
			Outage probability~$\probruin_{t}$ for different initial budgets~$\initbudget$ over time for a system using the deterministic scheme.
			The channels to both Bob and Eve are Rayleigh fading with average \glspl{snr} $\expect{\X_i}=\SI{20}{\dB}$ and $\expect{\Y_i}=\SI{10}{\dB}$, respectively.
			The solid lines correspond to the numerically calculated outage probabilities according to \eqref{eq:prob-surv-recursive} while the markers indicate results from \gls{mc} simulations with $10^6$ samples (\autoref{ex:fixed-time-rayleigh}).
		}
		\label{fig:outage-prob-rayleigh-fixed-time}
	\end{figure}
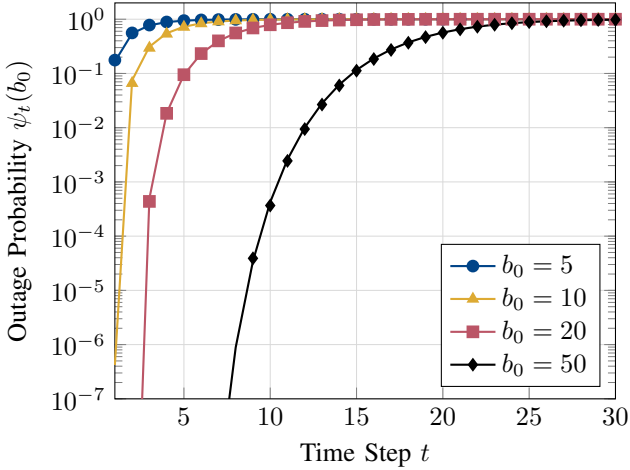
\end{example}

{%
\begin{rem}
	Throughout this work, we use the above Rayleigh fading example to numerically illustrate the theoretical results and show how they can be applied.
	However, it should be noted that these general results can be applied to a variety of scenarios, even including those with an active jammer.
	In this case, the distributions of the channel gains~$\X_i$ and $\Y_i$, and thus the resulting distributions of~$\income_i$ and $\claim_i$ need to be adjusted accordingly.
	The source code for all the presented results and calculations is made freely available at \cite{BesserGithub} and can be adapted to custom numerical examples, e.g., for other types of fading like Rician fading.
\end{rem}
}

\subsection{Latency Analysis}
Having analyzed the reliability of the communication system in the {active state}, we next investigate the latency between two {active states}.
Recall from \autoref{sec:system-model} that we assume that the system enters a {recharge phase} once the \gls{sk} bit budget is exhausted, cf.~\autoref{fig:model-budget-illustration}.
In this phase, only \gls{skg} is performed until a minimum number of available key bits is reached.
The number of time slots this process takes is defined to be the latency~$\latency$ of the system.

The minimum initial budget that needs to be reached depends on the requirements of the application.
In the following, we assume that there is a constraint on the outage probability for a given time slot, i.e., the system must remain in the {active state} for $\survtime$~time slots with an outage probability of at most~$\varepsilon$.
Equivalently, this means that the system should survive at least $\survtime$ time slots with a probability of at least $1-\varepsilon$.
Thus, the minimum initial budget to achieve this is
\begin{equation*}
	\initbudget^{\survtime}(\varepsilon) = \inf \left\{b\geq 0 \;\middle|\; 1-\varepsilon \leq \probsurv_{\survtime}(b)\right\}\,,
\end{equation*}
which corresponds to the inverse of $\probsurv_{\survtime}$, i.e.,
\begin{equation}\label{eq:init-budget-inverse-ruin-prob}
	\initbudget^{\survtime}(\varepsilon) = \inv{\probsurv_{\survtime}}(1-\varepsilon) = \inv{\probruin_{\survtime}}(\varepsilon)\,.
\end{equation}

The operational meaning behind this is that we have to generate at least $\initbudget^{\survtime}(\varepsilon)$ \gls{sk} bits, before we can start transmitting.
Once the \gls{sk} bit budget is exhausted, this process needs to be repeated.
Therefore, this quantity determines the latency before the communication system can securely transmit messages again.

Since the \gls{skg} rates in each time slot vary randomly, the latency~$\latency$ is also a random variable.
The following theorem characterizes the average latency for a system with \gls{iid} channel realizations.

\begin{thm}[{Average Latency}]\label{thm:average-latency-fixed-time}
	Consider the described communication system with a tolerated outage probability~$\varepsilon$ for a specified survival duration~$\survtime$.
	The \gls{skg} rates~$\income_i$ are \gls{iid} with positive and finite expectation~$\expect{\income_1}$.
	The average latency between two {active system states} is given by
	\begin{equation}\label{eq:average-latency-fixed}
		\expect{\latency} = \frac{\initbudget^{\survtime}(\varepsilon)}{\expect{\income_1}}\,.
	\end{equation}
\end{thm}
\begin{proof}
	The proof can be found in \autoref{app:proof-thm-avg-latency}.
\end{proof}
\begin{rem}
	The (average) latency in \eqref{eq:average-latency-fixed} is defined in terms of independent channel realizations of the main channel.
	According to the timing model in \autoref{fig:model-fixed-time-scheme}, two independent channel realizations form one time slot.
	Thus, the latency in terms of time slots is half of that in \eqref{eq:average-latency-fixed}.
\end{rem}
\begin{rem}
	Up to the last step of the proof of \autoref{thm:average-latency-fixed-time}, we do not require $\initbudget$ to be constant.
	Therefore, the result only needs to be slightly modified in the case where $\initbudget$ is a random variable.
	In this case, we get $\expect{\latency} = \expect{\expect{\latency \;\middle|\; \initbudget}} = \expect{\initbudget}/\expect{\income_1}$.
\end{rem}

\begin{example}[Reliability-Latency Tradeoff]\label{ex:fixed-time-init-budget}
	We illustrate the result of the latency analysis with the following numerical example.
	Similar to \autoref{ex:fixed-time-rayleigh}, we assume \gls{iid} Rayleigh fading with average \glspl{snr} of $\expect{\X_i}=\SI{20}{\dB}$ and $\expect{\Y_i}=\SI{10}{\dB}$ for Bob's and Eve's channel gains, respectively.
	For this scenario, the minimum required initial budget~$\initbudget^{\survtime}(\varepsilon)$ is depicted in \autoref{fig:minimum-budget-latency-fixed} over the outage probability constraint~$\varepsilon$ for multiple time constraints~$\survtime$.
	The solid lines again show the numerically calculated values by the recursive relation in \eqref{eq:prob-surv-recursive}, while the markers indicate the results of \gls{mc} simulations with $10^6$ samples.
	
	The first clear observation is that the required initial budget drops to zero as the tolerated outage probability approaches one.
	Similarly, it increases with a stricter reliability constraint.
	Interestingly, for the chosen example of Rayleigh fading, the required $\initbudget^{\survtime}(\varepsilon)$ increases only slowly when decreasing~$\varepsilon$ at a fixed~$\survtime$.
	For $\survtime=5$, an initial budget of around $\initbudget^{5}(10^{-1})=\SI{19.9}{\bit}$ is required to survive with an outage probability less than $10^{-1}$.
	For a stricter requirement of $10^{-5}$, this increases to around $\initbudget^{5}(10^{-5})=\SI{33.3}{\bit}$, i.e., for a \num{10000}-times increase in reliability, the initial budget only needs to be increased by a factor of around \num{1.67}.
	
	In contrast, when increasing the duration~$\survtime$ that the system should remain in the {active state}, the increase in $\initbudget^{\survtime}(\varepsilon)$ is more significant.
	In order to survive $\survtime=10$ time slots with an outage probability of at most $\varepsilon=10^{-1}$, the required initial budget is around $\initbudget^{10}(10^{-1})=\SI{35.8}{\bit}$, which is an \num{1.8}-times increase compared to $\survtime=5$ with the same outage probability constraint.
	
	According to \autoref{thm:average-latency-fixed-time}, the average latency~$\expect{\latency}$ is a simple linear scaling of the required initial budget~$\initbudget^{\survtime}(\varepsilon)$.
	For this numerical example, we have an average \gls{skg} rate around $\expect{\income_1}=\SI{3.31}{\bit}$.
	Thus, we have an average latency of around $\expect{\latency}=6$~time slots for $\survtime=5$ and $\varepsilon=10^{{-1}}$.
	Similarly, this increases to around \num{10}~time slots for $\varepsilon=10^{-5}$.
	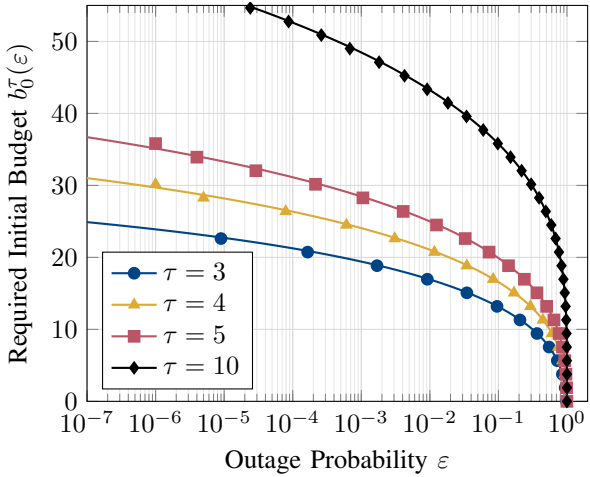
\begin{figure}
		\centering
		\begin{tikzpicture}%
	\begin{axis}[
		betterplot,
		ylabel={Required Initial Budget $\initbudget^{\survtime}(\varepsilon)$},
		xlabel={Outage Probability $\varepsilon$},
		ymin=0,
		ymax=55,
		xmode=log,
		xmin=1e-7,
		xmax=2,
		legend pos=south west,
		]
		
		\addplot+[only marks, mark repeat=5, forget plot] table[y=budget,x=t3] {data/ruin-cdf-budget-mc.dat};
		\addplot+[no marks, legend image post style={mark=*}] table[y=budget,x=t3] {data/ruin-cdf-budget-th.dat};
		\addlegendentry{$\tau=3$};
		
		\addplot+[only marks, mark repeat=5, forget plot] table[y=budget,x=t4] {data/ruin-cdf-budget-mc.dat};
		\addplot+[no marks, legend image post style={mark=triangle*}] table[y=budget,x=t4] {data/ruin-cdf-budget-th.dat};
		\addlegendentry{$\tau=4$};
		
		\addplot+[only marks, mark repeat=5, forget plot] table[y=budget,x=t5] {data/ruin-cdf-budget-mc.dat};
		\addplot+[no marks, legend image post style={mark=square*}] table[y=budget,x=t5] {data/ruin-cdf-budget-th.dat};
		\addlegendentry{$\tau=5$};
		
		\addplot+[only marks, mark repeat=5, forget plot] table[y=budget,x=t10] {data/ruin-cdf-budget-mc.dat};
		\addplot+[no marks, legend image post style={mark=diamond*}] table[y=budget,x=t10] {data/ruin-cdf-budget-th.dat};
		\addlegendentry{$\tau=10$};
	\end{axis}
\end{tikzpicture}
		\caption{%
			Required initial budget~$\initbudget^{\survtime}$ over the outage probability~$\varepsilon$ for a system in the {recharge state}.
			The channels to both Bob and Eve are Rayleigh fading with average \glspl{snr} $\expect{\X_i}=\SI{20}{\dB}$ and $\expect{\Y_i}=\SI{10}{\dB}$, respectively.
			The solid lines correspond to the numerically calculated outage probabilities according to \eqref{eq:init-budget-inverse-ruin-prob} while the markers indicate results from \gls{mc} simulations with $10^6$ samples	(\autoref{ex:fixed-time-init-budget}).
		}
		\label{fig:minimum-budget-latency-fixed}
	\end{figure}
\end{example}
\section{Random Transmission Times}\label{sec:random-arrival-time-scheme}
After having introduced and analyzed a deterministic timing scheme in the previous section, we now consider a different scheme in which messages arrive randomly.

\subsection{Random Transmission -- Model Description}
In each time slot~$t$, there is a probability~$\probtx_t$ that a message needs to be securely transmitted.
{An example where this timing model is applicable is a communication scenario between a sensor and a fusion center where the sensor only transmits data when an external event occurs, e.g., when the temperature rises above a threshold.
Assuming that these external events occur randomly, the times at which data is transmitted are also random.
By adjusting the threshold at which the sensor transmits data, the system designer could influence the transmission probability~$\probtx_t$.}

When a message needs to be transmitted, available \gls{sk} bits are used to encrypt the data.
Similar to the scheme in \autoref{sec:fixed-time-scheme}, this removes $\claim_t$~bits from the budget of available key bits.
However, if no transmission occurs, the time slot is used for \gls{skg}, which adds $\income_t$ bits to the budget.
An illustration can be found in \autoref{fig:model-random-time-scheme}.
Overall, the number of available \gls{sk} bits~$\budget_t$ at time~$t$ is again given by
\begin{equation}\label{eq:def-budget-random}
	\budget_t = \initbudget - \sum_{i=1}^{t} \netclaim_i = \initbudget - \netsum_{t}\,,
\end{equation}
where $\netclaim_i$ again denotes the net usage of key bits in time slot~$i$.
While this looks identical to the budget expression for the deterministic scheme in \eqref{eq:def-budget-fixed}, the distribution of the net usage~$\netclaim_i$ is different.
In contrast to the deterministic scheme, we now express $\netclaim_i$ as
\begin{equation}\label{eq:def-net-claim-random}
	\netclaim_i = \txindicator_i \claim_i - (1-\txindicator_i) \income_i\,,
\end{equation}
where $\txindicator_i\sim\berndist(\probtx_i)$ is an independent Bernoulli-distributed random variable indicating whether time slot~$i$ is used to transmit a message or perform \gls{skg}.
The distribution of $\netclaim_i$ then follows as
\begin{align}
	F_{\netclaim_i}(z)
	&= \Pr\left(\txindicator_i\claim_i  - (1-\txindicator_i)\income_i \leq z\right)\notag\\
	&= (1-\probtx_i)\Pr\left(-\income_i\leq z\right) + \probtx_i \Pr\left(\claim_i \leq z\right)\notag\\
	&= (1-\probtx_i)\bar{F}_{\income_i}(-z) + \probtx_i F_{\claim_i}(z)\,,
	\label{eq:cdf-net-claim-random}
\end{align}
where $\bar{F}_{\income_i}=1-F_{\income_i}$ denotes the survival function of~$\income_i$.
As before, we will assume throughout the following, that all involved random variables are mutually independent and identically distributed over time.

\begin{figure}
	\centering
	\begin{tikzpicture}%
	\begin{axis}[
		scale only axis,
		width=.95\linewidth,
		height=.13\textheight,
		axis lines=middle,
		xlabel={Time Slot~$t$},
		xlabel near ticks,
		x tick label as interval,
		ylabel={},
		yticklabels={},
		y axis line style={draw=none},
		ymajorticks=false,
		xmin=0.8,
		xmax=10.2,
		xtick distance=1,
		ymin=0,
		ymax=1.2,
		axis on top,
		]
		\draw[plot2,thick,fill=plot2,fill opacity=.3] foreach \n in {1, ..., 4, 6, 7, 9}{(\n,0) rectangle (\n+1,1)};
		\draw foreach \n in {1, ..., 4, 6, 7, 9}{(\n+0.5, 1) node[anchor=south] {SKG}};
		
		\draw[black,thick,dashed] foreach \n in {1, ..., 15}{(\n, 0) -- (\n, 2)};
		
		\draw[plot3,thick,fill=plot3,fill opacity=.3] foreach \n in {5, 8, 10}{(\n,0) rectangle (\n+1,1)};
		\draw foreach \n in {5,8,10}{(\n+0.5, 1) node[anchor=south] {TX}};
	\end{axis}
\end{tikzpicture}
	\caption{%
		Illustration of the scheduling model with random \gls{tx} blocks {during the active system state}.
		In each time slot~$t$, there is a probability of~$\probtx_t$ that a message is transmitted.
		If no message is transmitted, \gls{skg} is performed instead.
	}
	\label{fig:model-random-time-scheme}
\end{figure}
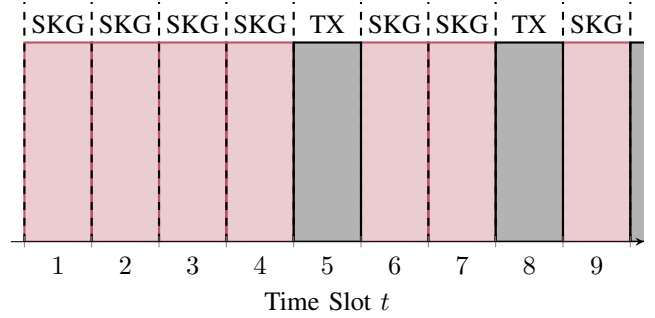

\begin{rem}
	For the expression of $\netclaim_i$ in \eqref{eq:def-net-claim-random} recall that $\income_i$ denotes the (non-negative) \gls{skg} rate.
	Since the generated key bits are added to the budget, the income~$\income_i$ is a negative usage.
\end{rem}
\begin{rem}
	The distribution of~$\netclaim_i$ in \eqref{eq:cdf-net-claim-random} has the form of a mixture distribution where $\txindicator_i$ is a mixture of $\claim_i$ and $-\income_i$.
\end{rem}

\subsection{Reliability Analysis}
Even though the distribution of~$\netclaim_i$ is different from that in \autoref{sec:fixed-time-scheme}, the progress of the budget~$\budget_t$ over time is still given by a random walk, cf.~\eqref{eq:def-budget-random}.
Thus, we can reuse the results on the outage probability~$\probruin_{t}$ from \autoref{sub:reliability-fixed} by adjusting the distribution of~$\netclaim_i$.
However, there exists a major qualitative difference when switching to the random time scheme.

In the deterministic scheme, the average net usage of \gls{sk} bits is greater than zero, i.e., $\expect{\netclaim_i}>0$, cf.~\autoref{lem:expected-net-claim}.
Therefore, on average, the system loses \gls{sk} bits, which results in an eventual outage, cf.~\autoref{cor:prob-ultimate-ruin-deterministic}.
In contrast, in the random scheme, it is possible for $\expect{\netclaim_i}$ to be negative, as shown in the following.

\begin{cor}[{Average Net Usage -- Random Scheme}]\label{cor:expected-net-claim-random}
	Consider the described communication system in the {active state} where a message is transmitted with probability~$\probtx_t$ in time slot~$t$.
	The following relation between the expected value of the net usage~$\netclaim_i$ and the transmission probability~$\probtx_i$ holds
	\begin{equation}\label{eq:relation-tx-prob-neg-expect-net-claim-random}
		\expect{\netclaim_i}\gtreqqless 0 \quad \Leftrightarrow \quad \probtx_i \gtreqqless \frac{\expect{\income_i}}{\expect{\income_i} + \expect{\claim_i}}\,.
	\end{equation}
\end{cor}
\begin{proof}
	The expected value of the net usage in time slot~$i$ is given as
	\begin{align*}
		\expect{\netclaim_i}
		&= \expect{\txindicator_i\claim_i} - \expect{(1-\txindicator_i)\income_i}\\
		&= \expect{\txindicator_i}\expect{\claim_i} - \expect{1-\txindicator_i}\expect{\income_i}\\
		&= \probtx_i\expect{\claim_i} - (1-\probtx_i)\expect{\income_i}
		\,,
	\end{align*}
	where we use the assumption that $\txindicator_i$ is independent from both $\income_i$ and $\claim_i$.
	This yields the relation
	\begin{equation*}
		\expect{\netclaim_i}\gtreqqless 0 \quad \Leftrightarrow \quad \probtx_i \gtreqqless \frac{\expect{\income_i}}{\expect{\income_i} + \expect{\claim_i}}\,,
	\end{equation*}
	which concludes the proof.
\end{proof}

The important consequence of \autoref{cor:expected-net-claim-random} is that, depending on~$\probtx_i$, the average net usage of \gls{sk} bits may be negative, i.e., on average more bits are generated than used in each time slot.
In contrast to the result in \autoref{cor:prob-ultimate-ruin-deterministic} for the deterministic scheme, this implies that there is a non-zero chance for the system to remain in the {active state} indefinitely.

Based on the calculation of the outage probability~$\probruin_{t}$ within a finite time horizon~$t$ in \eqref{eq:prob-surv-recursive}, we can calculate the probability of ultimate ruin~$\probruin = \lim\limits_{t\to\infty}\probruin_{t}$ according to the integral equation
\begin{equation}\label{eq:prob-ultimate-ruin-random-integral}
	\probruin(b) = (1-F_{\netclaim}(b)) + \int_{0}^{\infty} \probruin(s) f_{\netclaim}(b-s)\diff{s}
\end{equation}
where we assume \gls{iid} $\netclaim_i$ with density~$f_{\netclaim}$.
This is a Fredholm integral equation of the second kind, which can be solved numerically, e.g., by applying Nyström's method~\cite[Chap.~19.1]{Press2007recipes}.
For the numerical examples presented in the following, this technique is used to approximate the exact solution.

While the relation from \eqref{eq:prob-ultimate-ruin-random-integral} allows an exact calculation of the outage probability~$\probruin$ in theory, a solution may be difficult to obtain in practice.
Therefore, we provide a worst-case bound in the following theorem, which may be easier to compute in practice.

\begin{thm}[{Upper Bound of Probability of Ultimate Ruin -- Random Scheme}]\label{thm:prob-ultimate-ruin-random}
	Consider the described communication scheme in the {active state} with \gls{iid} \gls{skg} rates~$\income_i$ and \gls{iid} rates to Bob~$\claim_i$.
	In time slot~$i$, a message is transmitted with probability~$\probtx_i$.
	This probability is the same for all time slots~$i$ and fulfills
	\begin{equation}
		\probtx_1 = \cdots{} = \probtx_i = \probtx < \frac{\expect{\income_1}}{\expect{\income_1} + \expect{\claim_1}}\,.
	\end{equation}
	In this case, the probability of eventually leaving the {active state}~$\probruin$ is upper bounded by
	\begin{equation}\label{eq:bound-prob-ult-ruin-random-lundberg}
		\probruin(\initbudget) \leq \exp(-\adjcoeff^\star\initbudget)
	\end{equation}
	with $\adjcoeff^\star$ being the positive solution to
	\begin{equation}\label{eq:def-adj-coeff-ult-ruin-prob-random}
		\expect{\exp\left(\adjcoeff^\star \netclaim_1\right)} = 1\,,
	\end{equation}
	assuming it exists.
\end{thm}
\begin{proof}
	The proof can be found in \autoref{app:proof-thm-up-bound-ult-outage-prob-random}.
\end{proof}

The upper bound from \eqref{eq:bound-prob-ult-ruin-random-lundberg} shows that there is a non-zero probability of staying in the {active system state} indefinitely if the expected net usage of \gls{sk} bits $\expect{\netclaim_i}$ is negative, i.e., more bits are generated than used on average.

\begin{rem}
	The above results have direct implications for practical system design.
	If the system designer can influence the transmission probability~$\probtx$, its value should be set below the critical value $\expect{\income_1}/(\expect{\income_1}+\expect{\claim_1})$.
	Similar to \autoref{sec:fixed-time-scheme}, we can then calculate the minimum required initial budget~$\initbudget^{\infty}(\varepsilon)$ for a tolerated outage probability~$\varepsilon$.
	Thus, with probability~$\varepsilon$ there will be only the initial latency to generate $\initbudget^{\infty}(\varepsilon)$ key bits and no subsequent {recharging phases}.
\end{rem}

\begin{example}[{Rayleigh Fading -- Random Scheme}]\label{ex:rayleigh-fading-random}
	For this numerical example, we assume the same system parameters as in the previous \autoref{ex:fixed-time-rayleigh} and \autoref{ex:fixed-time-init-budget}, i.e., \gls{iid} Rayleigh fading with $\expect{\X_i}=\SI{20}{\dB}$ and $\expect{\Y_i}=\SI{10}{\dB}$.
	This gives us the expected values of the \gls{skg} rates~$\income_i$ and Bob's rates~$\claim_i$ as $\expect{\income_i}=\SI{3.31}{\bit}$ and $\expect{\claim_i}=\SI{5.889}{\bit}$, respectively.
	From \eqref{eq:relation-tx-prob-neg-expect-net-claim-random} in \autoref{cor:expected-net-claim-random}, the critical transmission probability is calculated to be $3.31/9.199=0.360$, i.e., for transmission probabilities~$\probtx<0.360$ the average net usage of \gls{sk} bits is negative.
	According to \autoref{thm:prob-ultimate-ruin-random}, this implies that the outage probability for such a scenario approaches a value less than \num{1} over time.
	
	The numerical results are presented in \autoref{fig:outage-prob-rayleigh-random}.
	First, we show the outage probability~$\probruin_{t}$ in finite time for two values of the transmission probability~$\probtx<0.36$ with an initial budget of $\initbudget=\SI{20}{\bit}$ in \autoref{fig:outage-prob-time-random}.
	The probabilities of ultimate ruin~$\probruin$, which are indicated by the dashed lines, are numerically evaluated according to \eqref{eq:prob-ultimate-ruin-random-integral}.
	Additionally, the upper bound from \eqref{eq:bound-prob-ult-ruin-random-lundberg} in \autoref{thm:prob-ultimate-ruin-random} is given by the dash-dotted lines.
	It can be seen that the outage probability~$\probruin_t$ approaches the limit~$\probruin$ for large~$t$.
	In the case of $p=0.1$, this limit is around $\probruin=1.45\cdot10^{-3}$ and the upper bound is around $3.53\cdot{}10^{-3}$.
	As expected, these values increase with an increasing transmission probability~$\probtx$, since more blocks are used for transmission which uses up more of the available \gls{sk} bits.
	For $p=0.35$, the probability of eventually exhausting the key-bit budget is $\probruin=0.64$, i.e., there is a $\SI{36}{\percent}$ chance that the system will remain in the {active state} indefinitely.
	
	Next, we show the behavior of the probability of ultimate ruin~$\probruin$ over the initially available number of \gls{sk} bits~$\initbudget$ in \autoref{fig:outage-prob-budget-random}.
	As expected, this probability decreases with an increasing initial budget and the slope is steeper for smaller transmission probabilities~$\probtx$.
	\begin{figure}
		\centering
		\subfloat[{Comparison of outage probability in finite time~$\probruin_{t}$ with the probability of ultimate ruin~$\probruin$ for an initial budget~$\initbudget=\SI{20}{\bit}$.}]{%
			\begin{tikzpicture}%
	\begin{axis}[
		betterplot,
		xlabel={Time Step $t$},
		ylabel={Outage Probability $\probruin_t(\initbudget)$},
		xmin=1,
		xmax=30,
		ymode=log,
		ymin=1e-5,
		ymax=2,
		legend pos=south east,
		]
		
		\addplot+ table[x=time,y=mc] {data/ult-ruin-prob-time-b20.0-p0.1.dat};
		\addlegendentry{$\probtx=0.1$};
		\addplot+ table[x=time,y=mc] {data/ult-ruin-prob-time-b20.0-p0.35.dat};
		\addlegendentry{$\probtx=0.35$};
		
		\pgfplotsset{cycle list shift=-2};
		\addplot+[no marks, dashed, very thick] table[x=time,y=th] {data/ult-ruin-prob-time-b20.0-p0.1.dat};
		\addplot+[no marks, dashed, very thick] table[x=time,y=th] {data/ult-ruin-prob-time-b20.0-p0.35.dat};
		
		\pgfplotsset{cycle list shift=-4};
		\addplot+[no marks, dashdotted, very thick] table[x=time,y=up] {data/ult-ruin-prob-time-b20.0-p0.1.dat};
		\addplot+[no marks, dashdotted, very thick] table[x=time,y=up] {data/ult-ruin-prob-time-b20.0-p0.35.dat};
	\end{axis}
\end{tikzpicture}
			\label{fig:outage-prob-time-random}
		}

		\subfloat[{Probability of ultimate ruin~$\probruin$ over the initial budget~$\initbudget$. The \gls{mc} results correspond to the values after $t=150$ time steps.}]{%
			\begin{tikzpicture}%
	\begin{axis}[
		betterplot,
		xlabel={Initial Budget~$\initbudget$},
		ylabel={Outage Probability $\probruin(\initbudget)$},
		xmin=1,
		xmax=30,
		ymode=log,
		ymin=1e-5,
		ymax=2,
		legend pos=south west,
		]
		
		\addplot+[no marks, dashed, legend image post style={mark=*}, very thick] table[x=budgetth,y=th] {data/ult-ruin-prob-budget-p0.1.dat};
		\addlegendentry{$\probtx=0.1$};
		\addplot+[no marks, dashed, legend image post style={mark=triangle*}, very thick] table[x=budgetth,y=th] {data/ult-ruin-prob-budget-p0.35.dat};
		\addlegendentry{$\probtx=0.35$};
		
		\pgfplotsset{cycle list shift=-2};
		\addplot+[only marks, mark repeat=4] table[x=budgetmc,y=mc] {data/ult-ruin-prob-budget-p0.1.dat};
		\addplot+[only marks, mark repeat=4] table[x=budgetmc,y=mc] {data/ult-ruin-prob-budget-p0.35.dat};
		
		\pgfplotsset{cycle list shift=-4};
		\addplot+[no marks, dashdotted, very thick] table[x=budgetup,y=up] {data/ult-ruin-prob-budget-p0.1.dat};
		\addplot+[no marks, dashdotted, very thick] table[x=budgetup,y=up] {data/ult-ruin-prob-budget-p0.35.dat};
	\end{axis}
\end{tikzpicture}
			\label{fig:outage-prob-budget-random}
		}
		\caption{%
			Probability of ultimate ruin~$\probruin$ for a system with transmission probability~$\probtx$.
			The channels to both Bob and Eve are Rayleigh fading with average \glspl{snr} $\expect{\X_i}=\SI{20}{\dB}$ and $\expect{\Y_i}=\SI{10}{\dB}$, respectively.
			The dashed lines correspond to the numerically calculated outage probabilities~$\probruin$ according to \eqref{eq:prob-ultimate-ruin-random-integral} while the markers indicate results from \gls{mc} simulations with $10^6$ samples.
			The dash-dotted lines represent the upper bound from \eqref{eq:bound-prob-ult-ruin-random-lundberg}.
			(\autoref{ex:rayleigh-fading-random})
		}
		\label{fig:outage-prob-rayleigh-random}
	\end{figure}
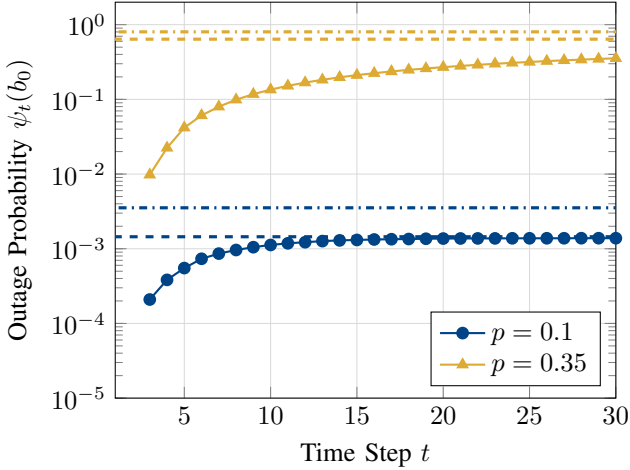
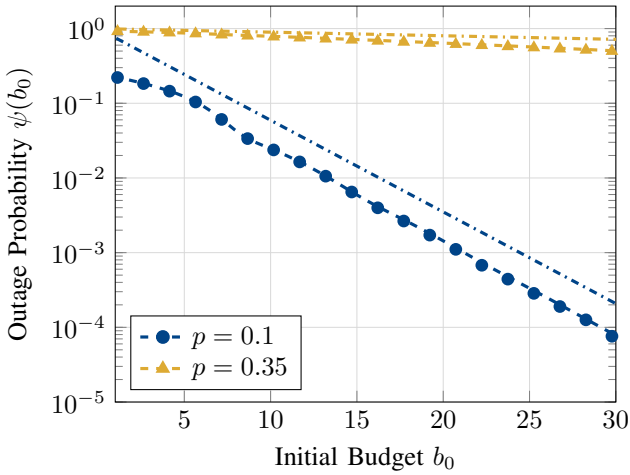
\end{example}
\section{Conclusion}\label{sec:conclusion}

In this work, we have considered a wireless communication system with a passive eavesdropper.
Alice and Bob perform \gls{skg} to generate key bits, which are added to a pool of available key bits.
When transmitting a message, bits from this pool are used as a one-time pad to secure the data transmission.
In this setting, we have analyzed the reliability in terms of the probability that the budget of available key bits will be exhausted.
We have shown how to compute this outage probability numerically and, additionally, derive worst-case bounds.
Interestingly, for randomly arriving messages, there is a positive probability that the communication system will remain active indefinitely and never run out of \gls{sk} bits.
In this case, there is only an initial latency to reach the initial budget of key bits before a transmission can take place.

For the cases where the system eventually runs out of available key bits, we have additionally investigated the latency between two active states, i.e., the duration during which Alice and Bob only generate new keys without performing any data transmission.
It is shown that the expected latency can be computed by a simple expression that is linear in the number of \gls{sk} bits that need to be initially available.

{System designers can utilize the findings of this work to adjust parameters, meeting specific performance requirements, e.g., the outage probability after a specified time.}

Since calculating the exact outage probability requires solving an integrodifference equation recursively, it will be interesting for future work to find approximations or tight bounds that are less computationally expensive.
In addition, the assumption that the channel distributions are stationary could be removed in future work.
In this way, a scenario with time-varying channel distributions could also be modeled and analyzed.

\appendices
\section{Proof of \autoref{lem:expected-net-claim}}\label{app:proof-lem-avg-net-usage-fixed}
Based on the definitions of $\netclaim_i$, $\claim_i$, and $\income_i$, we find that
\begin{align*}
	\expect{\netclaim_i}
	&= \expect{\claim_i - \income_i}\\
	\begin{split}
		&= \expect{\log_2\left(1+\Y_i\right)} + \expect{\log_2\left(1+\Xt_i\right)} \\&\hspace*{10em} - \expect{\log_2\left(1+\X_i+\Y_i\right)}
	\end{split}\\
	\begin{split}
		&= \expect{\log_2\left((1+\Y_i)(1+\Xt_i)\right)} \\&\hspace*{10em} - \expect{\log_2\left(1+\X_i+\Y_i\right)}
	\end{split}\\
	\begin{split}
		&= \expect{\log_2\left(1+\Xt_i+\Y_i+\Xt_i\Y_i\right)} \\& \hspace*{10em} - \expect{\log_2\left(1+\X_i+\Y_i\right)}
	\end{split}\\
	\begin{split}
		&= \expect{\log_2\left(1+\X_i+\Y_i+\X_i\Y_i\right)} \\&\hspace*{10em} - \expect{\log_2\left(1+\X_i+\Y_i\right)}
	\end{split}\\
	&> 0\,,
\end{align*}
where we require $\Xt_i\simdist\X_i$ and $\X_i, \Y_i > 0$.
\section{Proof of \autoref{thm:bound-upper-ruin-prob-fix}}\label{app:proof-thm-up-bound-outage-prob-fixed}
For the first inequality $\probruin_{t}\leq\boundupper_t$, we use Doob's martingale inequality~\cite[{Chap.~12.6, Thm.~1}]{Grimmett2001}.
In order to be able to apply it, we first formulate $\probruin_t$ in terms of the probability of the maximum of $\netsum_i$ as follows.
We start with the survival probability~$\probsurv_{t}$, which describes the probability that the \gls{sk} budget~$\budget_t$ never falls to zero up to time~$t$, i.e.,
\begin{align*}
	\probsurv_{t}(\initbudget) = \Pr\left(\min_{1\leq i \leq t} \budget_t > 0\right)\,.
\end{align*}
This can be rewritten using the definition of $\budget_t$ from \eqref{eq:def-budget-fixed} as
\begin{align*}
	\probsurv_{t}(\initbudget) = \Pr\left(\max_{1\leq i \leq t} \netsum_i < \initbudget\right)\,,
\end{align*}
which in turn can be expressed in terms of the outage probability~$\probruin_{t}$ as
\begin{equation*}
	\probruin_{t}(\initbudget) = \Pr\left(\max_{1\leq i \leq t} \netsum_i \geq \initbudget\right)\,.
\end{equation*}
Next, we need to show that the random process $\netsum_t=\sum_{i=1}^{t}\netclaim_i$ forms a submartingale (together with the natural filtration), i.e., we need to show that
\begin{equation*}
	\expect{\netsum_{t+1} \;\middle|\; \netsum_{t}, \dots, \netsum_{1}} \geq \netsum_t
\end{equation*}
holds.
This result can be obtained as follows:
\begin{align*}
	\mathbb{E}\Big[\netsum_{t+1} \;\Big|\; \netsum_{t}, &\dots, \netsum_{1}\Big]\\
	&= \expect{\netclaim_1 + \cdots{} + \netclaim_t + \netclaim_{t+1} \;\middle|\; \netclaim_1, \dots{}, \netclaim_t}\\
	&= \netclaim_1 + \cdots{} + \netclaim_t + \expect{\netclaim_{t+1} \;\middle|\; \netclaim_1, \dots{}, \netclaim_t}\\
	&\overset{(a)}{=} \netsum_{t} + \expect{\netclaim_{t+1}}\\
	&\overset{(b)}{\geq} \netsum_{t}\,,
\end{align*}
where $(a)$ holds due to the assumption that all $\netclaim_i$ are independent from each other and $(b)$ follows from \autoref{lem:expected-net-claim}.
The inequality $\probruin_{t}\leq\boundupper_t$ now follows directly from Doob's martingale inequality.

For the second inequality, we use the following general observations.
For a real-valued random variable~$\rv{X}$ with the expected value $\mu=\expect{\rv{X}}$ and variance $\var(\rv{X})=\expect{\rv{X}^2}-\mu^2$, it is apparent that
\begin{equation*}
	0 \leq \expect{\max(\rv{X}, 0)} \leq \expect{\abs*{\rv{X}}}\,.
\end{equation*}
From Jensen's inequality, it follows that
\begin{equation*}
	\expect{\abs*{\rv{X}}}^2 < \expect{\abs*{\rv{X}}^2}
\end{equation*}
since $x^2$ is a strictly convex function.
Combining the above with the fact that $\expect{\abs*{\rv{X}}^2} = \expect{{\rv{X}}^2}$ and $\expect{{\rv{X}}^2}=\var(\rv{X})+\mu^2$, we find that
\begin{align*}
	\expect{\max(\rv{X}, 0)}^2 
	&\leq \expect{\abs*{\rv{X}}}^2\\
	& < \expect{\abs*{\rv{X}}^2}\\
	& = \expect{{\rv{X}}^2}\\
	& = \var(\rv{X})+\mu^2\,,
\end{align*}
and in turn
\begin{equation*}
	\expect{\max(\rv{X}, 0)} < \sqrt{\var(\rv{X})+\mu^2}\,.
\end{equation*}
Now we can set $\rv{X}=\netsum_t$ to obtain the second inequality $\boundupper_t<\hat{\boundupper}_t$.
\section{Proof of \autoref{thm:average-latency-fixed-time}}\label{app:proof-thm-avg-latency}
By definition of the latency~$\latency$ in \eqref{eq:def-latency}, it can be seen that it is a hitting time of the random walk~$\sum_{i=1}^{t}\income_i$.
In \cite[Thm.~4 and 8]{Kotzing2019}, lower and upper bounds on the expected value of such hitting times for random walks with drift are derived.
In particular, it is shown that for a random walk~$\left\{\rv{A}_t\right\}$, the following equality holds:
\begin{equation*}
	\expect{\rv{T} \;\middle|\; \rv{A}_0}=\frac{\rv{A}_0}{\delta}\,,
\end{equation*}
with
\begin{equation*}
	\rv{T} = \inf \left\{t\geq 0 \;\middle|\; \rv{A}_t \leq 0\right\}\,,
\end{equation*}
if $\rv{A}_t-\expect{\rv{A}_{t+1}\;\middle|\; \rv{A}_0, \dots{}, \rv{A}_t} = \delta$ holds.

In order to use this result, we use the correspondence $\rv{A}_t = \initbudget - \sum_{i=1}^{t} \income_i$ with $\rv{A}_0=\initbudget$.
With this, we obtain
\begin{align*}
	\mathbb{E}\Big[ \rv{A}_{t+1}\;&\Big|\; \rv{A}_0, \dots{}, \rv{A}_t \Big]\\
	&= \expect{\initbudget - \income_1 - \income_2 - \cdots{} - \income_t - \income_{t+1} \;\middle|\; \rv{A}_0, \dots{}, \rv{A}_t}\\
	&= \initbudget - \income_1 - \income_2 - \cdots{} - \income_t - \expect{\income_{t+1}}\\
	&= \rv{A}_t - \expect{\income_{t+1}}\,,
\end{align*}
where the steps closely follow the ones in the proof of \autoref{thm:bound-upper-ruin-prob-fix}.
Based on the above relation, it is clear that
\begin{equation*}
	\rv{A}_t-\expect{\rv{A}_{t+1}\;\middle|\; \rv{A}_0, \dots{}, \rv{A}_t}
	= \expect{\income_{t+1}}
	= \delta\,.
\end{equation*}
Therefore, we have based on \cite{Kotzing2019}
\begin{equation*}
	\expect{\latency \;\middle|\; \initbudget} = \frac{\initbudget}{\expect{\income_1}}\,,
\end{equation*}
where $\expect{\income_1} = \cdots = \expect{\income_{t+1}}$ stems from the fact that we assume all $\income_i$ to be \gls{iid}.
Since $\initbudget$ is assumed to be a constant, it follows for the expectation of~$\latency$ that
\begin{equation*}
	\expect{\latency} = \expect{\expect{\latency \;\middle|\; \initbudget}} = \frac{\initbudget}{\expect{\income_1}}\,.
\end{equation*}
Since we need to reach $\initbudget^{\survtime}(\varepsilon)$ as an initial budget for the given system parameters, we obtain \eqref{eq:average-latency-fixed} as a final expression.
\section{Proof of \autoref{thm:prob-ultimate-ruin-random}}\label{app:proof-thm-up-bound-ult-outage-prob-random}
The proof closely follows the lines of the proof of Lundberg's inequality for the classical insurance risk model~\cite[Chap.~5]{Schmidli2017}.

First, recall that we aim to show
\begin{equation*}
	\probruin(\initbudget) = \lim\limits_{t\to\infty} \probruin_{t}(\initbudget) \leq \exp\left(-\adjcoeff^\star \initbudget\right)\,.
\end{equation*}
We start the proof by introducing the following functions:
\begin{align}
	g(\adjcoeff) &= \log\left(\expect{\exp\left(\adjcoeff\netclaim_1\right)}\right)\,,\\
	\rv{A}_t^{\adjcoeff} &= \exp\left(\adjcoeff\netsum_{t} - tg(\adjcoeff)\right)\,,
\end{align}
where $\netsum_{t}=\sum_{i=1}^{t}\netclaim_i$ is again the accumulated net usage.
Next, we show that the stochastic process $\{\rv{A}_t^{\adjcoeff}\}_{t\geq 1}$ is a martingale,
\begin{align*}
	\mathbb{E}\Big[ \rv{A}_{t+1}^{\adjcoeff} \;\Big|\;& \rv{A}_t^{\adjcoeff}, \dots{}, \rv{A}_{1}^{\adjcoeff}\Big]\\
	&= \expect{\exp\left(\adjcoeff\netsum_{t+1} - (t+1)g(\adjcoeff)\right) \;\middle|\; \rv{A}_t^{\adjcoeff}, \dots{}, \rv{A}_{1}^{\adjcoeff}}\\
	&= \expect{\rv{A}_{t}^{\adjcoeff} \cdot \exp\left(\adjcoeff\netclaim_{t+1} - g(\adjcoeff)\right) \;\middle|\; \rv{A}_t^{\adjcoeff}, \dots{}, \rv{A}_{1}^{\adjcoeff}}\\
	&= \rv{A}_{t}^{\adjcoeff} \cdot \expect{\exp\left(\adjcoeff\netclaim_{t+1} - g(\adjcoeff)\right) \;\middle|\; \rv{A}_t^{\adjcoeff}, \dots{}, \rv{A}_{1}^{\adjcoeff}}\\
	&\stackrel{(a)}{=} \rv{A}_{t}^{\adjcoeff} \cdot \expect{\exp\left(\adjcoeff\netclaim_{t+1} - g(\adjcoeff)\right)}\\
	&= \rv{A}_{t}^{\adjcoeff} \cdot \expect{\exp\left(\adjcoeff\netclaim_{t+1} - \log\left(\expect{\exp\left(\adjcoeff\netclaim_1\right)}\right)\right)}\\
	&= \rv{A}_{t}^{\adjcoeff} \cdot \frac{\expect{\exp\left(\adjcoeff\netclaim_{t+1}\right)}}{\expect{\exp\left(\adjcoeff\netclaim_1\right)}}\\
	&\stackrel{(b)}{=} \rv{A}_{t}^{\adjcoeff}\,,
\end{align*}
where we use the independence assumption in step~$(a)$ and the assumption of identical distributions in step~$(b)$.
In order to take the outage condition into account, we define the following stopping time:
\begin{equation}
	\tau_m = \inf\left\{t \;\middle|\; \netsum_{t} \geq \initbudget\right\} \wedge m\,,
\end{equation}
where $\wedge$ denotes the minimum operator.
By the definition of~$\adjcoeff^\star$ from \eqref{eq:def-adj-coeff-ult-ruin-prob-random}, we have that
\begin{equation*}
	\expect{\rv{A}_{1}^{\adjcoeff^\star}} = \expect{\exp\left(\adjcoeff^\star \netclaim_1\right)} = 1\,.
\end{equation*}
It further follows
\begin{align*}
	1
	&= \expect{\rv{A}_{1}^{\adjcoeff^\star}}\\
	&\stackrel{(a)}{=} \expect{\rv{A}_{\tau_m}^{\adjcoeff^\star}}\\
	&\geq \expect{\rv{A}_{\tau_m}^{\adjcoeff^\star} \one_{\tau_m < m}}\\
	&\stackrel{(b)}{\geq} \expect{\exp\left(\adjcoeff^\star\initbudget\right) \one_{\tau_m < m}}\\
	&= \exp\left(\adjcoeff^\star\initbudget\right) \Pr\left(\inf\left\{t \;\middle|\; \netsum_{t} \geq \initbudget\right\} < m\right)\\
	&\stackrel{(c)}{=} \exp\left(\adjcoeff^\star\initbudget\right) \probruin_{m}(\initbudget)\,,
\end{align*}
where $(a)$ is due to the optional stopping theorem~\cite[Chap.~12.5]{Grimmett2001}, $(b)$ follows from the definition of the stopping time~$\tau_m$, and $(c)$ uses the definition of the outage probability in finite time.
Rearranging finally yields
\begin{equation}
	\probruin_{m}(\initbudget) \leq \exp\left(-\adjcoeff^\star \initbudget\right), \quad \text{for all}\quad m\in\naturals\,.
\end{equation}
As a final step, we need to show that $\adjcoeff^\star$ is the positive solution to \eqref{eq:def-adj-coeff-ult-ruin-prob-random}, assuming it exists.
By Jensen's inequality, we have
\begin{equation*}
	g(\adjcoeff) = \log\left(\expect{\exp\left(\adjcoeff\netclaim_1\right)}\right) \geq \expect{\adjcoeff\netclaim_1}\,,
\end{equation*}
and by the definition of $\adjcoeff^\star$, it follows that
\begin{equation*}
	g(\adjcoeff^\star) = 0 \geq \adjcoeff^\star \expect{\netclaim_i}\,,
\end{equation*}
which is a contradiction for $\adjcoeff^\star < 0$, since $\expect{\netclaim_i} < 0$ by the assumption on~$\probtx$ and \autoref{cor:expected-net-claim-random}.

\printbibliography[heading=bibintoc]

\end{document}